\documentclass[11pt]{amsart}

\usepackage{microtype}

\usepackage[english]{babel}

\usepackage[hmarginratio=1:1,top=32mm,columnsep=20pt]{geometry}
\usepackage[font=it]{caption}
\usepackage{paralist}
\usepackage{multicol}

\usepackage[flushleft]{threeparttable}
\usepackage{graphicx}

\usepackage{lingmacros}
\usepackage{tree-dvips}

\usepackage{booktabs} 
\usepackage{mathtools}
\usepackage{amsthm}
\usepackage{amsfonts}
\usepackage{amsmath}
\usepackage{amssymb}
\usepackage{bm}
\usepackage{enumitem}
\usepackage{calligra}

\usepackage{todonotes}

\usepackage{algorithm}
\usepackage{algorithmicx}
\usepackage{algpseudocode}

\theoremstyle{definition}

\newtheorem{theo}{Theorem}
\newtheorem{prop}{Proposition}
\newtheorem{lemma}{Lemma}

\newcommand{\R}{\mathbb{R}}

\newcommand{\algr}{\mathcal{R}}
\newcommand{\algrbar}{\overline{\mathcal{R}}}
\newcommand{\algrbard}{{\overline{\mathcal{R}}}_{\mathcal{D}}}

\newcommand{\floor}[1]{\left\lfloor#1\right\rfloor}
\newcommand{\ceil}[1]{\left\lceil#1\right\rceil}

\newcommand{\Prob}[1]{\Pr\left[#1\right]}
\newcommand{\bigo}[1]{\mathcal{O}\left(#1\right)}
\newcommand{\bigomega}[1]{\Omega\left(#1\right)}
\newcommand{\set}[1]{\left\{#1\right\}}
\newcommand{\abs}[1]{\left|#1\right|}

\newcommand{\Def}{:=}

\newcommand{\e}[2]{\left(#1,#2\right)}

\DeclareMathOperator{\Exp}{Exp}

\newcounter{proofsection}

\begin{document}

\title{Randomization and Quantization for Average Consensus
}

\author{Bernadette Charron-Bost and Patrick Lambein-Monette}
\thanks{The authors are with the computer science laboratory at the \'Ecole
	polytechnique, Palaiseau, France.
}
\thanks{E-mails: \texttt{\{charron,patrick\}@lix.polytechnique.fr}
}

\maketitle

\begin{abstract}
A variety of problems in distributed control involve a networked system of
	autonomous agents cooperating to carry out some complex task in a
	decentralized fashion, e.g., orienting a flock of drones, or aggregating data
	from a network of sensors.
Many of these complex tasks reduce to the computation of a global function of
	values privately held by the agents, such as the maximum or the average.
Distributed algorithms implementing these functions should rely on limited
	assumptions on the topology of the network or the information available to the
	agents, reflecting the decentralized nature of the problem.

We present a randomized algorithm for computing the average in networks with
	directed, time-varying communication topologies.
With high probability, the system converges to an estimate of the average in
	linear time in the number of agents, provided that the communication
	topology remains strongly connected over time.
This algorithm leverages properties of exponential random variables, which
	allows for approximating sums by computing minima.
It is completely decentralized, in the sense that it does not rely on agent
	identifiers, or global information of any kind.
Besides, the agents do not need to know their out-degree; hence, our
	algorithm demonstrates how randomization can be used to circumvent the
	impossibility result established in~\cite{DBLP:conf/allerton/HendrickxT15}.

Using a logarithmic rounding rule, we show that this algorithm can be used under
	the additional constraints of finite memory and channel capacity.
We furthermore extend the algorithm with a termination test, by which the
	agents can \emph{decide irrevocably} in finite time\,---\,rather than simply
	converge\,---\,on an estimate of the average.
This terminating variant works under asynchronous starts and yields linear
	decision times while still using quantized\,---\,albeit larger\,---\,values.
\end{abstract}

\section{Introduction}
The subject of this paper is the \emph{average consensus} problem.
We fix a finite set~$V$ of~$n$ autonomous agents.
Each agent~$u$ has a scalar input value~$\theta_u \in \R$, and all agents
	cooperate to estimate, within some error bound~$\varepsilon$, the
	average~$\theta \Def (1/n) \sum_{u \in V} \theta_u$ of the inputs.
They do so by maintaining a local variable~$x_u$, which they drive
	close to the average~$\theta$ by exchanging messages with neighboring agents.
An algorithm achieves average consensus if, in each of its
	executions, each variable~$x_u$ gets sufficiently close to the
	average~$\theta$, namely $x_u \in [\theta-\varepsilon, \theta+\varepsilon]$,
	after a finite number of computation steps.

The study of this problem is motivated by a wide array of practical distributed
	applications, which either directly reduce to computing the average of
	well-chosen values, like sensor fusion~\cite{DBLP:conf/ipsn/XiaoBL05,
		DBLP:journals/jsac/CarliCSZ08, DBLP:journals/siamco/FagnaniFR14},
	or use average computation as a key subroutine, like load
	balancing~\cite{DBLP:journals/tac/NedicOOT09,
		DBLP:journals/mst/MuthukrishnanGS98}.
Other examples of such applications include
	formation control~\cite{DBLP:journals/pieee/Olfati-SaberFM07},
	distributed optimization~\cite{DBLP:journals/tac/NedicOOT09}, and task
	assignment~\cite{DBLP:journals/trob/ChoiBH09}.

\subsection{Contribution}
In this paper, our focus is on the design of efficient algorithms for average
	consensus.
Specifically, we are concerned with the \emph{convergence time}, defined as
	the number of communication phases needed to get each~$x_u$ within the 
	range~$[\theta-\varepsilon, \theta+\varepsilon]$.
It clearly depends on various parameters, including the input
	values, the error bound~$\varepsilon$, connectivity properties of the network,
	and the number~$n$ of agents.

The main contribution of this paper is a linear time algorithm in~$n$ that
	achieves average consensus in a networked system of anonymous agents, i.e.,
	without identifiers, with a time-varying topology that is continuously
	strongly connected.
It is a \emph{Monte Carlo} algorithm in the sense that the agents make use of
	private random oracles and may compute a wrong estimate of the average, but
	with a typically small probability.
We do not assume any stability or bidirectionality of the communication links,
	nor do we provide the agents with any global knowledge (like a bound on the
	size of the network) or knowledge of the number of their out-neighbors.
We also show how, by adding an initial quantization phase, we can make the
	memory and bandwidth requirements grow with~$n$ only as~$\log\log\,n$.

This is to be considered in the light of the impossibility result stated
  in~\cite{DBLP:conf/allerton/HendrickxT15}: deterministic, anonymous agents
	communicating by broadcast and without knowledge of their out-neighbors cannot
	compute functions dependent on the multiplicity or the order of the input
	values.
In particular, computing the average in this context requires providing the
  agents with knowledge of their out-degrees (or some equivalent information),
  or centralized control in the form of agent identifiers.
In contrast, our algorithm shows that, using randomization, computing the
	average can be done in a purely decentralized fashion, without using the
	out-degrees, even on a time-varying communication topology.

\subsection{Related works}

Average consensus is a specific case of the general \emph{consensus} problem,
	where the agents only need to agree on any value in the range of the input
	values.
Natural candidates to solve this general problem are \emph{convex combination
	algorithms}, where each agent repeatedly broadcasts its latest
	estimate~$x_u(t-1)$, and then picks~$x_u(t)$ in the range of its incoming
	values.
For agent~$u$ at time~$t$, this takes the form of the update rule
\begin{align}
	x_u(t) \leftarrow \sum_{v \in V} a_{uv}(t)x_v(t-1) \label{eq:dynfsto}
\end{align}
	where~$u$ sets the weight~$a_{uv}(t)$ to~$0$ if it has not received
	agent~$v$'s value at round~$t$.
To specify a convex combination algorithm amounts to describing how each
	agent~$u$ selects the positive weights~$a_{uv}(t)$.
The evolution of the system is then determined by the initial values~$\theta_u$
	and the stochastic weight
	matrices~$A(t) \Def {\left[a_{uv}(t)\right]}_{u, v \in V}$.

Convex combination algorithms have been extensively studied; see
	e.g.,~\cite{DBLP:phd/ndltd/Tsitsiklis84,
		DBLP:journals/tac/NedicOOT09,
		DBLP:journals/corr/Olshevsky14a,
		DBLP:journals/tac/JadbabaieLM03,
		DBLP:journals/tac/Moreau05,
		DBLP:journals/tac/Olfati-SaberM04,
		DBLP:journals/siamrev/OlshevskyT11,
		DBLP:journals/scl/XiaoB04,
		DBLP:conf/icalp/Charron-BostFN15,
	DBLP:journals/mst/MuthukrishnanGS98}.
The estimates~$x_u(t)$ have been shown to converge to a common limit under
	various assumptions of agent influence and network
	connectivity~\cite{DBLP:journals/siamco/CaoMA08a, DBLP:journals/tac/Moreau05}.
Unfortunately, convex combination algorithms suffer from poor time complexity in
	general: as shown in~\cite{DBLP:journals/tac/OlshevskyT13,
	DBLP:conf/icalp/Charron-BostFN15}, they may exhibit an exponentially large
	convergence time, even on a fixed communication topology.

However, some specific convex combination algorithms are known to converge in
	polynomial time, e.g.,~\cite{DBLP:journals/siamrev/OlshevskyT11,
	DBLP:journals/tac/NedicOOT09,
	DBLP:journals/siamco/Chazelle11}.
They all have in common that the Perron eigenvectors of the weight matrices are
	constant, which indeed is shown in~\cite{Charron-BostUnpublished} to guarantee
	convergence in~$\bigo{n^2/\alpha}$ if there exists a positive lower
	bound~$\alpha$ on all positive weights and if the network is continuously
	strongly connected.
Polynomial bounds are essentially optimal, as no convex combination algorithm
	can achieve convergence at an earlier time than $\bigomega{n^2}$ on every
	topology~\cite{DBLP:journals/tac/OlshevskyT11}.

While convex combination algorithms achieve asymptotic consensus with a positive
	lower bound on positive weights and high enough connectivity, the limit is
	only guaranteed to be in the range of the initial values and may be different
	from the average~$\theta$.
For example, the linear-time consensus algorithm
	in~\cite{DBLP:conf/icalp/Charron-BostFN16} works over any dynamic topology
	that is continuously rooted, but it converges to a value that is not equal
	to~$\theta$ in general.
In contrast, the convex combination algorithm
	in~\cite{DBLP:journals/tac/NedicOOT09} computes the average by selecting
	weights such that all the weight matrices~$A(t)$ are doubly stochastic.
To ensure this condition, agents need to collect some non local informations
	about the network, which requires some link stability over time, namely a
	three-round link stability.

In~\cite{DBLP:journals/mst/MuthukrishnanGS98} and subsequent works, 
	agents enrich the update rule~(\ref{eq:dynfsto}) with a second order term:
	\[
		x_u(t) \leftarrow
			\beta \sum_{v \in V} a_{uv}(t)x_v(t-1) + (1 - \beta) x_u(t-2).
	\]
The parameter~$\beta$ is usually a function of the spectral values of the
 	communication graph, which are hard to compute in a distributed fashion.
A notable exception is the algorithm proposed by Olshevsky
	in~\cite{DBLP:journals/corr/Olshevsky14a}, where the weights are locally
	computable, and the second order factor only depends on a bound~$N$ on the
	size~$n$ of the network.
Unfortunately, the latter algorithm assumes a fixed bidirectional topology and
	an initial consensus since all agents must agree on the bound~$N$.
Moreover, its time complexity is linear in~$N$, which may be arbitrarily large.

Other quadratic or linear-time average consensus algorithms elaborating on the
	update rule~(\ref{eq:dynfsto}) have been proposed in~\cite{sundaram2007finite,
		DBLP:journals/automatica/YuanSSBG13, DBLP:conf/amcc/Kibangou12,
	DBLP:journals/automatica/HendrickxJOV14}.
All of these actually solve a stronger problem, in that they achieve consensus
	on the exact average~$\theta$ in finite time, but they are computationally
	intensive and highly sensitive to numerical imprecisions.
Moreover, they are designed in the context of a fixed topology and some
	centralized control.

Average consensus algorithms built around the update rule~(\ref{eq:dynfsto})
	typically require bidirectional communication links with some stability and
	assume that agents have access to global information.
This is to be expected, as they operate by broadcast and over anonymous
	networks, and thus have to bypass the impossibility result
	in~\cite{DBLP:conf/allerton/HendrickxT15}: they do so through the use, at
	least implicitly, of the out-degree of the agents.

Another example is to be found in the \emph{Push-Sum}
	algorithm~\cite{DBLP:conf/focs/KempeDG03} in which agents make explicit use of
	their out-degrees in the messages they send.
This method converges on fixed strongly connected
	graphs~\cite{DBLP:conf/cdc/Dominguez-GarciaH11a}, and on continuously
	strongly connected dynamic graphs~\cite{DBLP:journals/corr/abs-1709-08765}.


Another way to circumvent the impossibility result
	in~\cite{DBLP:conf/allerton/HendrickxT15} consists in assuming unique agent
	identifiers: by tagging each initial value~$\theta_u$ with $u$'s identifier,
	at each step of the \emph{Flooding} algorithm, the agents can compute the
	average of the input values that they have heard of so far, and thus compute
	the global average~$\theta$ after~$n -1$ communication steps when the topology
	is continuously strongly connected.
Unfortunately, the price to pay in this simple average consensus algorithm is
	messages of size in~$\bigo{n\log n}$ bits. 
By repeated leader election on the basis of agent identifiers and a shared
	bound on the network diameter, the quadratic time algorithm
	in~\cite{DBLP:journals/tcns/OlivaSH17} also achieves average consensus, using
	message and memory size in only~$\bigo{\log n}$ bits with a fixed, strongly
	connected network.

Our approach is dramatically different from the above sketched ones:
	equipping each agent with private random oracles enables them to estimate
	the average with neither central control nor global information.
In particular, our algorithm requires no global clock, and agents may start
	asynchronously.
Communication links are no more assumed to be bidirectional and they may change
	arbitrarily over time, provided that the network remains permanently strongly
	connected.

Our algorithm leverages the fact that the minimum function can be easily
	computed in this general setting.
By individually sampling exponential random distributions with adequate rates
	and then by computing the minimum of the so-generated random numbers, agents
	can estimate the sum of initial values and the size of the network, yielding
	an estimate of the average.
This approach was first introduced in~\cite{DBLP:conf/podc/Mosk-AoyamaS06} for a
	gossip communication model, and later applied
	in~\cite{DBLP:conf/stoc/KuhnLO10} to the design of distributed counting
	algorithms in networked systems equipped with a global clock
	that delivers synchronous start signals to the agents.
	
The main features of some of the average consensus algorithms discussed above,
	including our own randomized algorithm, denoted~$\algrbar$, are summarized in
	Table~1.

\begin{table}
	\caption{Average consensus algorithms with continuous strong connectivity}\label{table:comparison}
	\centering
    \begin{threeparttable}
      \begin{tabular}{p{2.2cm} c c p{7cm}}
        \toprule

        Algorithm
        & Time
        & Message size
        & Restrictions

        \\
        \midrule

        Flooding *
        & $\bigo{n}$
        & $\bigo{n \log n}$
        & non anonymous network

         \\
        Ref.~\cite{DBLP:journals/tac/NedicOOT09}
        & $\bigo{n^2 }$
        & $\bigo{1 }$
        & bidirectional topology \newline link stability over three rounds

         \\
        Ref.~\cite{DBLP:journals/automatica/YuanSSBG13} *
        & $\bigo{n}$
        & $\infty$
        & fixed and bidirectional topology \newline computationally intensive

         \\
        Ref.~\cite{DBLP:journals/corr/Olshevsky14a}
        & $\bigo{N }$
        & $\infty$
        & fixed and bidirectional topology \newline $N \geq n$, known by all agents 

          \\
        Ref.~\cite{DBLP:journals/tcns/OlivaSH17} *
        & $\bigo{nD}$
        & $\bigo{\log n}$
        & fixed topology $G$ \newline $D \geq diam(G)$, known by all agents
					\newline non anonymous network

         \\
	Algorithm~$\overline{\mathcal{R}} $
        & $\bigo{ n }$
        & $\bigo{\log\log n } $
        & Monte Carlo algorithm

        \\
        \bottomrule
      \end{tabular}
      \begin{tablenotes}
        \small
        \item * These algorithms compute the exact average
      \end{tablenotes}
    \end{threeparttable}
\end{table}

\subsection{Quantization}

Most average consensus algorithms require agents to store and transmit real
	numbers.
This assumption is unrealistic in digital systems, where agents have finite
	memory and communication channels have finite capacity.
These constraints entail agents to use only quantized values.

Convex combination algorithms are not, in general, robust to quantization.
However, those that compute the average using doubly stochastic influence
	matrices have been shown to degrade gracefully under several specific rounding
	schemes, either deterministic~\cite{DBLP:journals/tac/NedicOOT09}, where the
	degradation induced by rounding is bounded, or
	randomized~\cite{DBLP:journals/tsp/AysalCR08}, where the expected average in
	the network is kept constant.

Other methods elaborating on the update rule~(\ref{eq:dynfsto}) have not, in
	general, been shown to behave well under rounding, like the second-order
	algorithm in~\cite{DBLP:journals/siamco/Olshevsky17}, or of the various
	protocols in~\cite{DBLP:conf/amcc/Kibangou12,
		DBLP:journals/automatica/HendrickxJOV14,
		sundaram2007finite,
	DBLP:journals/automatica/YuanSSBG13}.
In this context, one important feature of our algorithm is that it can be
	adapted to work with quantized values, following a logarithmic rounding scheme
	similar to the one in~\cite{Oshman12}.
With this rounding rule, each quantized value can be represented using
	$\bigo{\log\log n}$ bits.

\subsection{Irrevocable decisions}

The specification of average consensus can be strengthened by requiring
	that agents \emph{irrevocably decide} on some good estimates of the
	average~$\theta$ in finite time.
In other words, agents are required to detect when consensus on~$\theta$ has
	been reached within a given error margin.
This is desirable for many applications, e.g., when the average consensus
	algorithm is to be used as a subroutine that returns an estimate of~$\theta$
	to the calling program.
Various decision strategies have been developed for fixed topologies,
	e.g.,~\cite{yadav2007distributed, DBLP:conf/isccsp/ManitaraH14,
	DBLP:conf/allerton/ManitaraH13}.

Here, we design a decision test that uses the approximate value of~$n$ computed
	on-line and that incorporates the randomized \emph{firing scheme} developed
	in~\cite{DBLP:conf/stacs/Charron-BostM18} to tolerate asynchronous starts.
In this way, we show that the agents may still safely decide in linear time, but
	at the cost of larger messages.
Moreover, it achieves exact consensus since all the agents decide
	on the \emph{same} estimate of~$\theta$.

\subsection*{Organization}
The rest of this paper is organized as follows:
	in Section~\ref{sec:preliminaries}, we introduce our computational model and
	present some preliminary technical lemmas;
	we present our main algorithm in Section~\ref{sec:algorithm}, its
	quantized version in Section~\ref{sec:quantization}, and its variant with
	decision tests in Section~\ref{sec:decision}.
	Finally, we present concluding remarks in Section~\ref{sec:conclusion}.
\section{Preliminaries}\label{sec:preliminaries}
\subsection{Computation model}
We consider a networked system with a finite set~$V$ of~$n$ agents, and
  assume a distributed computational model in the spirit of the Heard-Of
  model~\cite{DBLP:journals/dc/Charron-BostS09}.
Computation proceeds in rounds: in a round, each agent broadcasts a message,
  receives messages from some agents, and finally updates its state according to
  some local algorithm.
Rounds are communication closed, in the sense that a message sent at round~$t$
  can only be received during round~$t$.
Communications that occur in a given round~$t$ are thus modeled by a directed
  graph~$G_t \Def \left(V,E_t\right)$: the edge~$\e{u}{v}$ is in~$E_t$ if
  agent~$v$ receives the message sent by agent~$u$ at round~$t$.
Any agent can communicate with itself instantaneously, so we assume a self-loop
  at each node in all of the graphs~$G_t$.

We consider \emph{randomized} distributed algorithms which execute an infinite
  sequence of rounds, and in which agents have access to private and independent
  random oracles.
Thus, an execution of a randomized algorithm is entirely determined by the
  collection of input values, the sequence of directed
  graphs~${\left(G_t\right)}_{t \geq 1}$, called a \emph{dynamic graph}, and the
  outputs of the random oracles.
We assume that the dynamic graph is managed by an \emph{oblivious adversary}
  that has no access to the outcomes of the random oracles.


We design algorithms to compute the average of initial values in a dynamic
  network.
Consider an algorithm where the local variable~$x_u$ is used to estimate the
  average.
We say that an execution of this algorithm
  \emph{$\varepsilon$-computes the average} if there is a round~$t^*$ such that,
  for all subsequent rounds~$t \geq t^*$, all estimates are within
  distance~$\varepsilon$ of the average~$\theta$ of the input values,
  namely~$x_u(t) \in \left[\theta-\varepsilon,\theta+\varepsilon\right]$ for
  all~$u \in V$.
The \emph{convergence time} of this execution is the smallest such round~$t^*$
  if it exists.

\subsection{Directed graphs and dynamic graphs}\label{sec:diygraph}

Let $G \Def (V,E)$ be a directed graph, with a finite set of nodes~$V$ of
  cardinality~$n$ and a set of edges~$E$.
There is a \emph{path} from~$u$ to~$v$ in~$G $ either if $\e{u}{v} \in E$, or if
  	$\e{u}{w} \in E$ and there is a path from~$w$ to~$v$.
If every pair of nodes is connected by a path, then~$G$ is said to be
  \emph{strongly connected}.
The dynamic graph~${\left(G_t\right)}_{t \geq 1}$ is said to be
  \emph{continuously strongly connected} if all the directed graphs~$G_t$ are strongly
  connected.

The \emph{product graph~$G \circ H$} of two directed
  graphs~$G \Def \left(V,E_G\right)$ and~$H \Def \left(V, E_H \right)$ is
  defined as $G \circ H \Def (V, E)$, with
  $E \Def \set{\e{u}{w} \in V \times V : \exists v \in V, \e{u}{v} \in E_G
  \wedge \e{v}{w} \in E_H} $.
Let us recall that the product of~$n-1$ directed graphs on~$V$ that are all
  strongly connected and have self-loops at each node is the complete graph.
It follows that, in every execution of the algorithm~\emph{Min}\,---\,given in
  Algorithm~\ref{alg:min}\,---\, over a continuously strongly connected dynamic graph, all
  agents have computed the smallest of the input values at the end of~$n-1$
  rounds.
The algorithm~\emph{Min} is a fundamental building block of our average
  consensus algorithms, and the latter observation will drive their convergence times.

\begin{algorithm}
  \caption{The algorithm \emph{Min}, code for agent~$u$}~\label{alg:min}
  \begin{algorithmic}[1]
    \State{\textbf{Input:} $\theta_u \in \R$
    }
    \State{$x_u \gets \theta_u$
    }\label{alg:min:init}
    \For{$t = 1, 2, \ldots$}
      \State{Send~$x_u$.
      }
      \State{Receive~$x_{v_1}, \ldots, x_{v_k}$ from neighbors.
      }
      \State{$x_u \gets \min \set{x_{v_1}, \ldots, x_{v_k}}$
      }\label{alg:min:update}
    \EndFor{}
  \end{algorithmic}
\end{algorithm}

\subsection{Exponential random variables}
For any positive real number~$\lambda$, we denote by $X \sim \Exp(\lambda)$
  that~$X$ is a random variable following an exponential distribution with
  rate~$\lambda$.
One easily verifies the following property of exponential random variables.

\begin{lemma}~\label{lemma:summin}
Let $X_1, \ldots, X_k$ be $k$ independent exponential random variables
	with rates $\lambda_1, \ldots, \lambda_k$, respectively.
Let~$X$ be the minimum of $X_1, \ldots, X_k$.
Then, $X$ follows an exponential distribution with
  rate~$\lambda \Def \sum_{i=1}^k \lambda_i$.
\end{lemma}

The accuracy of our algorithm depends on some parameter~$\ell$ whose value is
  determined by the bound in the following lemma, which is an application of the
  Cram\'er-Chernoff method (see for instance~\cite{Boucheron13}, sections 2.2 and
  2.4).

\begin{lemma}\label{lemma:cramer}
Let $X_1, \ldots, X_{\ell}$ be $\ell$ i.i.d.\ exponential random variables with
  rate $\lambda > 0$, and let $\alpha \in (0,1/2)$.
Then,
  \[
    \Prob{\abs{\frac{X_1 + \cdots + X_{\ell}}{\ell} - \frac{1}{\lambda}}
          \geq \frac{\alpha}{\lambda}}
      \leq 2 \exp \left( - \frac{ \ell \alpha^2}{3}\right) \enspace.
  \]
\end{lemma}

\section{Randomized algorithm}\label{sec:algorithm}

In this section, we assume infinite bandwidth channels and infinite storage
  capabilities.
For this model, we present a randomized algorithm~$\algr$ and prove that
  all agents  compute the same value which, with high probability, is a good estimate  
  of the average of the initial values.

The underlying idea is that each agent computes an estimate of~$s$, the sum of
	the input values, and an estimate of~$n$, the size of the network.
They use the ratio of the two estimates  as an estimate of the
  average~$\theta \Def s/n$.

The computations of the estimates of~$s$ and~$n$ are based on
	Lemma~\ref{lemma:summin}: each agent~$u$ samples two random numbers from two
  exponential distributions, with respective rates~$\theta_u$ and~$1$.
Then, the agent~$u$ computes the two global minima of these so-generated random
	numbers in the variables~$X_u$ and~$Y_u$ with the algorithm~\emph{Min}.
As recalled in Section~\ref{sec:diygraph}, this takes at most~$n-1$ rounds when
  the dynamic graph is continuously strongly connected.
Then, $1/X_u$ and $1/Y_u$ provide estimates of respectively~$s$ and~$n$.

The probabilistic analysis requires all  the input values to be at least equal to one.
To overcome this limitation, we assume that the agents know some
	pre-defined interval $[a,b]$ in which all the input values lie and 
	we apply a reduction to the case $a=1$ by simple translations of the inputs. 

We then elaborate on the above algorithmic scheme to decrease the probability of
  incorrect executions, i.e., executions with errors in the estimates that are
  greater than~$\varepsilon$.
We replicate each random variable~$\ell$ times, and each node starts with the
  two vectors $X_u = (X_u^{(1)}, \ldots, X_u^{(\ell)})$ and
  $Y_u =(Y_u^{(1)}, \ldots, Y_u^{(\ell)})$, instead of the sole
  variables~$X_u$ and~$Y_u$.
Using the Cram\'er-Chernoff bound given in Lemma~\ref{lemma:cramer}, we choose
  the parameter~$\ell$ in terms of the maximal admissible error~$\varepsilon$,
  the probability~$\eta$ of incorrect executions, and the amplitude~$b-a$ of the
  input values; namely, we set 
  $\ell \Def \ceil{27 \ln(4/\eta) {(b-a+1)}^2/\varepsilon^2}$.

 The pseudocode of the algorithm~$\algr$ is given in Algorithm~\ref{alg:r}.
\begin{algorithm}
  \caption{The algorithm $\algr$, code for agent~$u$}~\label{alg:r}
  \begin{algorithmic}[1]
    \State{\textbf{Input:} $\theta_u \in \left[a,b\right]$
    }
    \State{$\ell \gets \ceil{27 \ln(4/\eta) {(b - a + 1)}^2/\varepsilon^2}$
    }
    \State{$x_u \gets \bot$
    }
    \State{Generate $\ell$ random numbers
      $\sigma_u^{(1)}, \ldots, \sigma_u^{(\ell)}$
      from an exponential distribution of rate~$\theta_u - a + 1$.
    }
    \State{$X_u \gets (\sigma_u^{(1)}, \ldots, \sigma_u^{(\ell)})$
    }
    \State{Generate $\ell$ random numbers $\nu_u^{(1)}, \ldots, \nu_u^{(\ell)}$
      from an exponential distribution of rate~$1$.
    }
    \State{$Y_u \gets (\nu_u^{(1)}, \ldots, \nu_u^{(\ell)})$
    }
    \State{\textbf{In each round} \textbf{do}}
    \State{Send $\left(X_u,Y_u\right)$.
    }
    \State{Receive $
      \left(X_{v_1},Y_{v_1}\right), \ldots,
      \left(X_{v_k},Y_{v_k}\right)
      $
      from neighbors.
    }
    \For{$i = 1 \ldots \ell$}
    \State{$X_u^{(i)} \gets \min \set{X_{v_1}^{(i)}, \ldots, X_{v_k}^{(i)}}$
    }
    \State{$Y_u^{(i)} \gets \min \set{Y_{v_1}^{(i)}, \ldots, Y_{v_k}^{(i)}}$
    }
    \EndFor{}
    \State{$x_u \gets a - 1 +
      (Y_u^{(1)} + \cdots + Y_u^{(\ell)})/(X_u^{(1)} + \cdots + X_u^{(\ell)})$
    }
  \end{algorithmic}
\end{algorithm}

\begin{theo}~\label{theo:r}
For any real numbers~$\varepsilon \in \left(0,1/2\right)$
	and~$\eta \in \left(0,1/2\right)$, in any continuously strongly connected network, the
  	algorithm~$\mathcal{R}$~$\varepsilon$-computes the average of initial values
		in $\left[a,b\right] $ in at most~$n-1$ rounds with probability at
		least~$1-\eta$.
\end{theo}

\begin{proof}
We first introduce some notation.
If~$z_u$ is any variable of node~$u$, we denote by~$z_u(t)$ the value of~$z_u$
	at the \emph{end} of round~$t$.
We let
  \[
    \hat{\sigma}^{(i)}  \Def \min_{u \in V} \sigma_u^{(i)},
    \qquad
    \hat{\nu}^{(i)}     \Def \min_{u \in V} \nu_u^{(i)},
  \]
  and
  \[
    \hat{\sigma}  \Def \sum_{i = 1}^{\ell} \hat{\sigma}^{(i)} / \ell,
    \qquad
    \hat{\nu}     \Def \sum_{i = 1}^{\ell} \hat{\nu}^{(i)}    / \ell,
    \qquad
    \hat{\theta}  \Def \frac{\hat{\nu}}{\hat{\sigma}} \enspace.
  \]

As an immediate consequence of the connectivity assumptions, for each node~$u$
	and each index~$i \in \set{1,\ldots,\ell}$, we
  have~$X_u^{(i)}(t) = \hat{\sigma}^{(i)}$ and~$Y_u^{(i)}(t) = \hat{\nu}^{(i)}$
  at every round~$t \geq n-1$.
Hence, $x_u(t) = a - 1 + \hat{\theta}$ whenever~$t \geq n-1$.

We now show that~$a - 1 + \hat{\theta}$ lies in the admissible range
  $\left[\theta-\varepsilon, \theta+\varepsilon\right]$ with probability at
  least~$1-\eta$.
By considering the translate initial values~$\theta'_u \Def \theta_u - a + 1$
  that all lie in~$\left[1,b-a+1\right]$, we obtain a reduction to the
  case~$a = 1$.

So let us assume that $a=1$.
In this case, $b$ is positive, and we let~$\alpha \Def \varepsilon/3b$.
Since $b \geq a = 1$ and~$\varepsilon \in \left(0,1/2\right)$, we
  have~$\alpha \in \left(0,1/6\right)$.
This implies $1 - 3\alpha < (1 - \alpha)/(1 + \alpha)$ and
  $1 + 3\alpha > (1 + \alpha)/(1 - \alpha)$.
It follows that, if~$\abs{\hat{\sigma} - 1/s}  \leq \alpha/s$ and
  $\abs{\hat{\nu} - 1/n}  \leq \alpha/n$, i.e.,
  \[
    \frac{1}{s} (1-\alpha)  \leq  \hat{\sigma}  \leq \frac{1}{s} (1 + \alpha)
    \ \mbox{ and } \
    \frac{1}{n} (1-\alpha)  \leq  \hat{\nu}     \leq \frac{1}{n} (1 + \alpha)
    \enspace ,
  \]
  then we have
  \[
    \frac{s}{n} (1 - 3\alpha)
    \leq \frac{ \hat{\nu} }{ \hat{\sigma} }
    \leq \frac{s}{n} (1 + 3 \alpha) \enspace ,
  \]
  i.e.,
  \[
    \abs{\hat{\theta} - \theta}  \leq   3 \alpha \theta
    \leq \varepsilon \enspace.
  \]

Specializing Lemma~\ref{lemma:cramer}
  for~$\ell \Def \ceil{27 \ln(4/\eta) b^2/\varepsilon^2}$
  and~$\alpha \Def \varepsilon/3b$, we get
  	\begin{align*}
    	\Prob{\abs{\frac{Z_1 + \cdots + Z_\ell}{\ell} - \frac{1}{\lambda}}
          \geq \frac{\alpha}{\lambda}}
         &  \leq 2 {\left(\frac{\eta}{4}\right)}^{{\left(\frac{3b\alpha}{\varepsilon}\right)}^2}
          = \frac{\eta}{2} \enspace ,
 	 \end{align*}
   where $Z_1, \dots, Z_{\ell}$ are i.i.d.\ exponential random variables of rate
   $\lambda > 0$.
In particular, $\Prob{\abs{\hat{\sigma} - 1/s} \geq \alpha/s} \leq \eta/2$ and
  $\Prob{\abs{\hat{\nu} - 1/n} \geq \alpha/n} \leq \eta/2$ since, by
  Lemma~\ref{lemma:summin}, we have $\hat{\sigma}^{(i)} \sim \Exp(s)$ and
  $\hat{\nu}^{(i)} \sim \Exp(n)$ with~$s$ and~$n$ that are both positive. 
The probability of the union of those two events  is thus less than~$\eta$.
Using the above argument and the fact that $\varepsilon/ b \leq \varepsilon$, we
  conclude that
  \[
    \Prob{\abs{\hat{\theta} - \theta}  \geq  \varepsilon} \leq \eta \enspace,
  \]
  which completes the proof.
\end{proof}

The convergence of the algorithm~$\algr$ in Theorem~\ref{theo:r} is ensured by the assumption of
	continuous strong connectivity of the dynamic graph~${\left(G_t\right)}_{t \geq 1}$:
	the directed graph~$G_1 \circ \cdots \circ \,G_{n-1}$ is complete, and thus the  entries~$Y_u^{(i)}$ and~$X_u^{(i)}$
	hold a global minimum at the end of round~$n-1$.
This connectivity  assumption may be dramatically reduced into \emph{eventual strong connectivity}:
	for each round~$t$, there exists a round~$t'$ such that~$G_t \circ \cdots \circ G_{t'}$ is the complete graph.
Clearly, the algorithm~$\algr$  converges with any dynamic graph that is eventually strongly connected, but
	the finite convergence time is then unbounded.

An intermediate connectivity assumption has been proposed in~\cite{DBLP:conf/stacs/Charron-BostM18}:
	a dynamic graph ${\left(G_t\right)}_{t \geq 1}$ is \emph{strongly connected with delay $T$} if each product of~$T$
	consecutive graphs~$G_t \circ \cdots \circ G_{t+T-1}$ is strongly connected.
Then, the convergence of the algorithm~$\algr$ is still guaranteed, but at the
	price of increasing the convergence time by a factor~$T$.

Conversely, the assumption of continuous strong connectivity can be strengthened in the following way:
	for any positive integer~$c$,
	a dynamic graph~${\left(G_t\right)}_{t \geq 1}$ is \emph{continuously} $c$-\emph{strongly connected} if
	each directed graph~$G_t$ is \emph{$c$-in-connected}, i.e., any non-empty subset~$S \subseteq V$
	has  at least~$\min\set{c, \abs{V\setminus S}}$ incoming neighbors in~$G_t$.
It can be shown that the product of~$\ceil{n/c}$ $c$-in-connected directed graphs is
	complete~\cite{DBLP:conf/stacs/Charron-BostM18}.
Hence, the assumption of continuous $c$-connectivity results in a speedup by a
	factor~$c$.

\section{Quantization}\label{sec:quantization}
In this section, we present a variant of the algorithm~$\algr$ that, as opposed
  to the former, works under the additional constraint that agents can only
  store and transmit quantized values.
 This model is intended for networked systems with communication bandwidth
  and storage limitations.
We incorporate this constraint in our randomized algorithm by requiring each
  agent~$u$ firstly to quantize the random numbers it generates, and secondly to
  broadcast only one entry of each of the two vectors~$X_u$
  and~$Y_u$ in each round.

The quantization scheme consists in rounding values down along a logarithmic
	scale, to the previous integer power of some pre-defined number greater than
	one. 
Exponential random variables, when rounded in this way, continue to follow
	concentration inequalities similar to those of Lemma~\ref{lemma:cramer}.
This makes logarithmic rounding appealing to use in conjunction with the
	algorithm~$\algr$, as we retain control over incorrect executions simply by
	increasing the number~$\ell$ of samples by a constant factor.

This quantization method does not offer an \emph{absolute} bound over the space
  and bandwidth potentially required in the algorithm:
	the generated random numbers may be arbitrarily large or small, and therefore
	the number of quantization levels used in all executions is unbounded.
Instead, we provide a \emph{probabilistic} bound over the number of quantization
	levels required\,---\,that is, a bound that holds with high probability.
All the random numbers that are generated lie in some pre-defined interval~$I$
  with high probability, and hence most executions of our algorithm require a
  pre-defined number~$Q$ of quantization levels.
In each of these ``good'' executions, random numbers can be represented
  efficiently, as~$Q$ grows with~$\log n$.

This probabilistic guarantee for quantization could be turned into an absolute
  one by providing the agents with a bound~$N \geq n$.
This is indeed the rounding scheme developed in~\cite{Oshman12}, 
  where each agent starts with normalizing the random numbers that it generates
  before rounding.
Our quantization method provides a weaker guarantee, but it does not use any
	global information about the network.
	
In the following, our quantized algorithm is denoted~$\algrbar$; its pseudocode
  is given in Algorithm~\ref{alg:rbar}.
It uses the rounding function~$r_\beta : x\in \R_{>0} \mapsto
r_\beta(x) \Def {\left(1 + \beta \right)}^{\floor{\log_{1 + \beta}x}}$, 
	where~$\beta$ is any positive number.
	
We start the correctness proof of~$\algrbar$ with a preliminary lemma that
	gives, for~$X \sim \Exp(\lambda)$, concentration inequalities for the
	logarithmically rounded exponential random variable~$r_\beta(X)$.

\begin{lemma}~\label{lemma:cramer-quant}
Let $X_1, \ldots, X_{\ell}$ be $\ell$ i.i.d.\ exponential random variables with
	rate $\lambda > 0$, and let~$\beta > 0$ and $\alpha \in (0,1/2)$.
Then,
  \[
    \Prob{\abs{\frac{r_\beta(X_1) + \cdots + r_\beta(X_\ell)}{\ell} - \frac{1}{\lambda}}
          \geq \frac{\alpha + \beta + \alpha\beta}{\lambda}}
      \leq 2 \exp \left( - \frac{\ell \alpha^2}{3}\right)
  \]
\end{lemma}
\begin{proof}
Let~$X \Def (X_1 + \cdots + X_\ell)/\ell$ and
  $Y \Def (r_\beta(X_1) + \cdots + r_\beta(X_\ell))/\ell$.
For any~$x > 0$, we have $r_\beta(x) \leq x < (1 + \beta)r_\beta(x)$,
	and hence
  \[
    0 \leq X - Y < \beta Y \leq \beta X \enspace.
  \]

It follows that if $\abs{X - 1/\lambda} \leq \alpha/\lambda$, then
  \begin{align*}
    \abs{Y - 1/\lambda}
    &\leq \abs{Y - X} + \abs{X - 1/\lambda}
    \leq \beta X + \alpha/\lambda\\
    &\leq \left(\alpha + \beta + \alpha\beta\right)/\lambda \enspace.
  \end{align*}
The result follows from the latter inequality and Lemma~\ref{lemma:cramer}.
\end{proof}

\begin{algorithm}
  \caption{The algorithm $\algrbar$, code for agent $u$}~\label{alg:rbar}
  \begin{algorithmic}[1]
    \State{\textbf{Input:} $\theta_u \in \left[a,b\right]$
    }
    \State{$\ell \gets \ceil{108 \ln(8/\eta){(b - a + 1)}^2/\varepsilon^2}$
    }
    \State{$\beta \gets \varepsilon/8(b-a+1)$
    }
    \State{$x_u \gets \bot$
    }
    \State{Generate~$\ell$ random numbers
      $\sigma_u^{(1)} \ldots, \sigma_u^{(\ell)}$ from an exponential
      distribution of rate~$\theta_u - a + 1$.
    }
    \State{}
      $X_u \gets (r_{\beta}(\sigma_u^{(1)}),
      \ldots, r_{\beta}(\sigma_u^{(\ell)}))$
    \State{Generate~$\ell$ random numbers $\nu_u^{(1)} \ldots, \nu_u^{(\ell)}$
      from an exponential distribution of rate~$1$.
    }
    \State{$Y_u \gets (r_{\beta}(\nu_u^{(1)}),
      \ldots, r_{\beta}(\nu_u^{(\ell)}))$
    }
    \State{$i \gets 0$}
    \State{\textbf{In each round}  \textbf{do}}
    \State{$i \gets i + 1$}
    \State{Send $\left(X_u^{(i)},Y_u^{(i)}\right)$.
    }
    \State{Receive $
      \left(X_{v_1}^{(i)}, Y_{v_1}^{(i)}\right), \ldots,
      \left(X_{v_k}^{(i)}, Y_{v_k}^{(i)}\right)
      $
      from neighbors.
    }
    \State{$X_u^{(i)} \gets \min \set{X_{v_1}^{(i)}, \ldots, X_{v_k}^{(i)}}$
    }
    \State{$Y_u^{(i)} \gets \min \set{Y_{v_1}^{(i)}, \ldots, Y_{v_k}^{(i)}}$
    }
    \If{$i = \ell$}
    \State{$x_u \gets a - 1 +
      (Y_u^{(1)} + \cdots + Y_u^{(\ell)})/(X_u^{(1)} + \cdots + X_u^{(\ell)})$
    }
    \State{$i \gets 0$}
    \EndIf{}
  \end{algorithmic}
\end{algorithm}

\begin{prop}~\label{prop:rbar}
For any real numbers~$\varepsilon \in \left(0,1/2\right)$
  and~$\eta \in \left(0,1/2\right)$, in any continuously strongly connected network, the
  algorithm~$\algrbar$ $\varepsilon$-computes the average of initial values that
  all lie in~$\left[a,b\right]$ with probability at least~$1 - \eta/2$ in at
  most~$\ell n$ rounds.
\end{prop}

\begin{proof}
  We let
  \[
    \dot{\sigma}^{(i)}  \Def \min_{u \in V} r_{\beta}(\sigma_u^{(i)}),
    \qquad
    \dot{\nu}^{(i)}     \Def \min_{u \in V} r_{\beta}(\nu_u^{(i)}),
  \]
  and
  \[
    \dot{\sigma}  \Def \sum_{i = 1}^{\ell} \dot{\sigma}^{(i)} / \ell,
    \qquad
    \dot{\nu}     \Def \sum_{i = 1}^{\ell} \dot{\nu}^{(i)}    / \ell,
    \qquad
    \dot{\theta}  \Def \frac{\dot{\nu}}{\dot{\sigma}} \enspace.
  \]

The main loop of the algorithm~$\algrbar$ consists in running many
	instances of the algorithm \emph{Min}, interleaving their executions so that
 	the variables~$X_u^{(i)}$ and~$Y_u^{(i)}$ are updated at
	rounds~$i, i + \ell, i + 2\ell, \ldots$
Since the topology is continuously strongly connected,
  $X_u^{(i)} (t) = \dot{\sigma}^{(i)}$ and~$Y_u^{(i)} (t) = \dot{\nu}^{(i)}$ for
  every round~$t \geq i + (n-1)\ell$.
Hence, $x_u(t) = a-1+\dot{\nu}/\dot{\sigma} = a-1+\dot{\theta}$
	whenever~$t \geq \ell n$.

Now we show that~$a-1+\dot{\theta}$ lies in the admissible range
  $[ \theta - \varepsilon, \theta + \varepsilon] $ with probability
	at least~$1-\eta/2$.
For that, we proceed as in Theorem~\ref{theo:r}: we reduce the general
	case to the case~$a=1$ by translation.

Since the function $r_\beta$ is non-decreasing, $\min$ and~$r_\beta$ commute.
Therefore, by Lemma~\ref{lemma:summin}, $\dot{\sigma}^{(i)}$
  and~$\dot{\nu}^{(i)}$ are the quantized values of two exponential random
  variables with respective rates~$s$ and~$n$.

We let~$\alpha \Def \varepsilon/6b$
  and~$\gamma \Def \alpha + \beta + \alpha\beta$.
Since~$b \geq a = 1$ and~$\varepsilon \in \left(0,1/2\right)$, we
  have~$0 < \gamma < \varepsilon/3b < 1/6$.
This implies that, if $\abs{\dot{\sigma} - 1/s} \leq \gamma/s$
  and~$\abs{\dot{\nu} - 1/n} \leq \gamma/n$, then
  	$\abs{\dot{\theta} - \theta} \leq 3\gamma\theta < \varepsilon$.

Using Lemma~\ref{lemma:cramer-quant}
  with~$\ell = \ceil{108 \ln(8/\eta) {(b-a+1)}^2/\varepsilon^2}$,
	$\alpha = \varepsilon/6(b-a+1)$, and $\beta = \varepsilon/8(b-a+1)$,
  we obtain $\Prob{\abs{\dot{\sigma} - 1/s} \geq \gamma/s} \leq \eta/4$ and
  $\Prob{\abs{\dot{\nu}    - 1/n} \geq \gamma/n} \leq \eta/4$.
  Therefore,
  \[
    \Prob{\abs{\dot{\theta} - \theta} \leq \varepsilon} \geq 1 - \eta/2
    \enspace.
  \]
\end{proof}

\begin{prop}~\label{prop:rbar-space}
For any real numbers~$\varepsilon \in \left(0,1/2\right)$ and~$\eta \in
	\left(0,1/2\right)$, in any continuously strongly connected network, each entry of the
  	vectors~$X_u$ and~$Y_u$ in algorithm~$\algrbar$ can be
    represented over~$Q = \bigo{\frac{1}{\varepsilon}
      \left(\log n - \log \eta - \log \varepsilon\right)}$
	quantization levels, with probability at least~$1-\eta/2$.
\end{prop}
\begin{proof}
If~$X \sim \Exp(\lambda)$ with~$\lambda \geq 1$, then for any~$z \in (0,1)$,
  \[
    \Prob{X \leq  z} = 1 - e^{-\lambda z} \leq \lambda z
    \ \text{and} \
    \Prob{X > \ln (1/z)} = z^\lambda \leq z \enspace.
  \]
Hence
  \[
    \Prob{X \notin \left[z,\ln 1/z\right]} \leq (1+\lambda) z \enspace.
  \]
In particular, when~$I$ denotes the interval $\left[z,\ln 1/z\right]$
  with $z = \frac{\eta}{4(b-a+2) \ell n} < \frac{1}{16}$, we obtain that, for
  each agent~$u$ and each index $i$,
  	$\Prob{\sigma_u^{(i)} \notin I} \leq \eta/4\ell n$ and
    $\Prob{\nu_u^{(i)}    \notin I} \leq \eta/4\ell n$.
Since the random numbers~$\sigma_u^{(i)}$ and~$\nu_u^{(i)}$ are all independent,
  we deduce that
  \[
    \Prob{\exists u \in V, i \in \set{1, \ldots, \ell} :
    \sigma_u^{(i)} \notin I \vee \nu_u^{(i)} \notin I} \leq \eta/2.
  \]

If all the random numbers~$\sigma_u^{(i)}$ and~$\nu_u^{(i)}$ lie in the
	interval~$\left[c,d\right] \subseteq \R_+$, then they are rounded
  	into the finite set~$r_{\beta}(\left[c,d\right])$, which means
    that~$Q \Def \abs{r_{\beta}\left([c,d]\right)}$ different quantization
    levels are sufficient to represent their logarithmically rounded values.
Since $\abs{r_{\beta}(\left[c,d\right])}
  \leq \ceil{\log_{1+\beta}(d)} - \floor{\log_{1+\beta}(c)}$, we have
  $Q = \bigo{\log_{1+\beta} \left(\ell n/\eta\right)}$ for the~$r_{\beta}$
  roundings of values in the interval 
	$\left[ \frac{\eta}{4(b-a+2) \ell n},
  \ln \left( \frac{4(b-a+2) \ell n}{\eta} \right) \right]$.
Observing that $\beta \in \left(0,1\right)$ and thus
  $\log_{1+\beta} x < 2\log x/\beta$,
	we have $Q = \bigo{(1/\beta)\log\left(\ell n/\eta\right)}$.
With the values of the parameters~$\beta$ and~$\ell$ as defined in the
  algorithm~$\algrbar$, lines 2 and 3, we finally obtain
  \[
    Q = \bigo{\frac{1}{\varepsilon}
    \left(\log n - \log \eta - \log \varepsilon \right)} \enspace.
  \]

\end{proof}

Combining Propositions~\ref{prop:rbar} and~\ref{prop:rbar-space}, we deduce the
  following correctness result for the algorithm~$\algrbar$.

\begin{theo}~\label{theo:rbar}
For any real numbers~$\varepsilon \in \left(0,1/2\right)$
  and~$\eta \in \left(0,1/2\right)$, in any continuously strongly connected network, the
  algorithm~$\algrbar$ $\varepsilon$-computes the average of initial values
  in~$\left[a,b\right]$ in at most~$\ell n$ rounds and using messages in
  $\bigo{\log\left(\log n - \log \eta\right) - \log \varepsilon}$ bits, with
  probability at least~$1 - \eta$.
\end{theo}

As above sketched, the algorithm~$\algrbar$ differs from~$\algr$ in several
  respects.
First, the length~$\ell$ of the random vectors is larger.
This is due to the fact that the concentration inequality in
  Lemma~\ref{lemma:cramer-quant} is looser than in Lemma~\ref{lemma:cramer}.
Moreover, we retain a safety margin of~$\eta/2$ for controlling executions in
  which some of the random numbers generated by the agents lie outside of the
  admissible interval for quantization.

Another discrepancy is that the agents send only one entry of each of the two
	vectors~$X_u$ and~$Y_u$ in each round of~$\algrbar$ while
	they send  the complete vectors in the algorithm~$\algr$.
This sequentialization implemented in the algorithm~$\algrbar$, results in 
  reducing the size of messages by a factor~$\ell$, but at the price of
  augmenting the convergence time by the same factor~$\ell$.

The use of this strategy also entails a stronger sensitivity on network
	connectivity than when broadcasting entire vectors at each round.
Indeed, the convergence of~$X_u^{(i)}$ and~$Y_u^{(i)}$ is now decorrelated from
	that of~$X_u^{(j)}$ and~$Y_u^{(j)}$ for~$j \neq i$.
Global convergence requires that for each index~$i$, the graph products of the form
  $ G_i \circ G_{i+\ell} \circ \cdots \circ G_{i+k.\ell} $ are all complete from some integer~$k$.
This condition is \emph{not} implied, for instance, by  continuous strong
  connectivity with delay~$T$, and indeed an adversary with knowledge of~$\ell$ can pick a
  dynamic graph that is $2$-delayed continuously strongly connected, and for which no
	progress is ever made for some entries of the vectors~$X_u$
  and~$Y_u$.

\section{Decision}\label{sec:decision}
So far, we have been concerned only with the convergence of each
	estimate~$x_u(t)$ to the average~$\theta$.
However, when used as a subroutine, an average consensus algorithm may have to
	return an estimate of the average~$\theta$ to the calling program.
In other words, the agents have to \emph{decide irrevocably} on an estimate of~$\theta$.

Formally, we equip each agent with a decision variable~$d_u$, initialized to~$\bot$.
Agent~$u$ is said to \emph{decide} the first time it writes in~$d_u$.
The corresponding problem is specified as follows:

\begin{description}
\item[Termination] $\forall u \in V,\  \exists t_u, \ \forall t \geq t_u, \ d_u(t) \neq \bot$. 
\item[Irrevocability] $\forall u \in V, \ \forall t \geq 1, \ \forall t'\geq t, \ d_u(t) = \bot \ \vee \  d_u(t) = d_u(t')$.
\item[Validity] $\forall u \in V, \ \forall t \geq 1, \ d_u(t) = \bot \ \vee \ d_u(t) \in \left[\theta - \varepsilon, \theta + \varepsilon\right]$.
 \end{description}

In this section, we seek to augment the algorithms~$\algr$ and $\algrbar$ to solve the above
	problem with high probability.
Our approach relies on the fact that in both algorithms, each agent converges in finite time.

A simple solution consists in providing the agents with a bound~$N \geq n$:
	each agent~$u$ stops executing~$\algr$ and decides at round~$N$.
From Theorem~\ref{theo:r}, it follows that termination, irrevocability,
	and validity hold with probability at least~$1-\eta$.
A similar scheme can be applied to the algorithm~$\algrbar$ and decisions at
	round~$ \ell N$.

Unfortunately, this approach suffers from two major drawbacks.
First, the time complexity of the resulting algorithms is arbitrarily large, as
	it depends on the quality of the bound~$N$.
Second, the decision tests involve the current round number~$t$, and hence
	require that the agents have access to this value, or at least start executing
	their code simultaneously.
Charron-Bost and Moran~\cite{DBLP:conf/stacs/Charron-BostM18} recently showed
	that synchronous starts can be emulated in continuously strongly connected networks,
	but at the price of a \emph{firing} phase of~$n$ additional rounds.

To circumvent the above two problems, we propose another approach that consists
	in using the estimate of~$n$ computed by the
	algorithms~$\algr$ and~$\algrbar$ in the decision tests, and in incorporating the randomized
	firing scheme developed in~\cite{DBLP:conf/stacs/Charron-BostM18} to
	tolerate asynchronous starts.
	
Let us briefly recall their model and techniques.
Each agent is initially \emph{passive}, i.e., it does nothing and
	emits only null messages (heartbeats).
Eventually it becomes \emph{active}, i.e., 
	it starts executing the algorithm.
An active agent~$u$ maintains a \emph{local virtual clock}
	$C_u$ with the following property: under the assumption of a dynamic network
	that is continuously strongly connected, the local clocks remain smaller than~$n$ as
	long as some agents are passive, and when all the agents are active, 
	they get synchronized to some value at most equal to~$n$.
Let~$s_{\max}$ denote the last round with passive agents.
At the end of round~$s_{\max}+n-1$, all agents have the same estimate~$n^*$
	of~$n$, which lies in~$\left[2n/3,3n/2\right]$ with high probability.
Hence, $C_u \geq 3 n^*/2 $ guarantees that $C_u \geq n$, and thus agent~$u$
	can safely decide.

The algorithm~$\algrbard$, given in Algorithm~\ref{alg:rbard}, integrates this
	decision mechanism in the algorithm~$\algr$, with the rounding of
	algorithm~$\algrbar$.

\begin{theo}
For any real numbers~$\varepsilon \in \left(0,1/2\right)$
	and~$\eta \in \left(0,1/2\right)$, in any continuously strongly connected network, with
	probability at least~$1-\eta$, the algorithm~$\algrbard$ decides on values
	within~$\varepsilon$ of the average of initial values in~$\left[a,b\right]$
	in~$s_{\max}+ 2n$ rounds, using messages
	in~$\bigo{\left(\log N - \log \eta / \varepsilon^2\right)
	\left(\log\left(\log N - \log \eta\right) - \log \varepsilon\right)}$ bits.
\end{theo}
\begin{proof}
We first observe that all the variables $X_u^{(i)}$, $Y_u^{(i)} $, and $n_u$ are
	stationary.
Since the dynamic graph is continuously strongly connected, their final values do not
	depend on agent~$u$.
Let $t_u^c$ be the first round from which all these variables are constant.
Section~\ref{sec:diygraph} shows that
	\begin{equation}\label{eq:tuc<}
	 t_u^c< s_{\max}+n \enspace.
	 \end{equation}
We let $d_u^* = d_u \big( t_u^c \big)$ and $n^* = n_u \big( t_u^c \big)$.

From~\cite{DBLP:conf/stacs/Charron-BostM18}, we know that the counters~$C_u$
		satisfy the following:
		\begin{enumerate}[label = (\roman*)]
		\item $\forall t \leq s_{\max}, \ \ C_u(t) < n$;
		\item $\exists t_0 \in \set{s_{\max}+1, \ldots, s_{\max}+n-1},
			\forall u\in V, \forall t \geq t_0, C_u(t) = t - s_{\max}$.
		\end{enumerate}
Since $n_u$ is upper bounded by $n^*$, the property (ii) entails that the
	agent~$u$ eventually decides. 
Hence, the termination property is ensured.
Let $t_u^d$ denote the first round at which the agent~$u$ decides,
		i.e., the first round such that
		\[
			C_u(t_u^d) > 3 \, n_u(t_u^d)/2 \enspace.
		\]

Observing that deciding in~$\algrbard$ coincides with firing in the randomized algorithm 
	in~\cite{DBLP:conf/stacs/Charron-BostM18}, the first part of the correctness proof of the latter algorithm shows that 
	\begin{equation} \label{eq:tud>n}
		\Prob{\forall u \in V, \ t_u^d \geq s_{\max} +n}~\geq~1-\eta/3 
	\end{equation}
	since $\ell \geq \ceil{243 \ln \left(6 N^2/\eta\right)}$.
Combined with~(\ref{eq:tuc<}), we obtain 
\[
	\Prob{\forall u \in V, \ t_u^d > t_u^c}~\geq~1 - \eta/3 \enspace.
\]
Because of the definition of $t_u^c$, decisions in~$\algrbard$ are thus irrevocable with probability at least $1 - \eta/3$.

The proof of the randomized firing algorithm also shows that if
	$\ell \geq \ceil{243 \ln \left(6 N^2/\eta\right)}$, then
	\begin{equation} \label{eq:tud>2n}
		\Prob{\forall u \in V, \ t_u^d \leq s_{\max} + 2 n}~\geq~1-\eta/3 \enspace.
	\end{equation}
In other words, all the agents decide by round $ s_{\max} + 2 n $ with probability at least $ 1 - \eta/3$.

Moreover, we observe that the computation of the estimate of $\theta$ in~$\algrbard$ corresponds to 
	the algorithm~$\algrbar$.
Then, Proposition~\ref{prop:rbar} shows that 
	\begin{equation}
		\Prob{\forall u \in V, \ d_u^* \in
			\left[ \theta-\varepsilon, \theta + \varepsilon \right]}~\geq~1-\eta/6
	\end{equation}
	since $\ell \geq \ceil{108\ln(24 /\eta){(b - a + 1)}^2/\varepsilon^2} $.
It follows that the validity property holds with probability at least $ 1 - \eta/6$.	

As opposed to~$\algrbar$, each agent~$u$ sends all the entries of $X_u$ and
	$Y_u$ in the messages of the algorithm~$\algrbard$. 
Moreover, the above argument shows that the agent~$u$ can stop sending $C_u$
	when it has decided.
Hence, in correct executions where the agents decide in linear time in~$n$, the
	counters~$C_u$ can be represented over~$\bigo{\log N}$ bits.

Reasoning as in Proposition~\ref{prop:rbar-space}, each entry of the
	vectors~$X_u$ and~$Y_u$ can be represented
	over~$Q = \bigo{\frac{1}{\varepsilon}\log\frac{\ell n}{\eta}}$ quantization
	levels with probability~$1-\eta/6$, where~$\ell = \bigo{\log N - \log
	\eta/\varepsilon^2}$.
Therefore, each message of~$\algrbard$ uses~$\bigo{\ell \log Q}$ bits with
	probability~$1-\eta/6$.
	
By the union bound over the latter four events, we obtain that all the agents in
	the algorithm~$\algrbard$ decide on values in the range
	$[\theta-\varepsilon, \theta + \varepsilon ]$ by round $s_{\max} + 2 n$ using
	messages of size in~$\bigo{\left(\log N - \log \eta / \varepsilon^2\right)
	\left(\log\left(\log N - \log \eta\right) - \log \varepsilon\right)}$ bits,
	with probability at least~$1-\eta$.
\end{proof}

\begin{algorithm}[h]
	\caption{The algorithm $\algrbard$, code for agent~$u$}~\label{alg:rbard}
	\begin{algorithmic}[1]
	\State{\textbf{Input:} $\theta_u \in \left[a,b\right]$
	}
	\State{$\ell \gets
		\max \set{ \ceil{108\ln(24 /\eta){(b - a + 1)}^2/\varepsilon^2} , \ceil{243 \ln \left(6N^2/\eta\right)}}$
	}
	\State{$\beta \gets \varepsilon/8(b-a+1)$
	}
	\State{Generate~$\ell$ random numbers
		$\sigma_u^{(1)} \ldots, \sigma_u^{(\ell)}$
		from an exponential distribution of rate~$\theta_u - a + 1$.
	}
	\State{$X_u \gets (r_{\beta}(\sigma_u^{(1)}),
		\ldots, r_{\beta}(\sigma_u^{(\ell)}))$
	}
	\State{Generate~$\ell$ random numbers $\nu_u^{(1)} \ldots, \nu_u^{(\ell)}$
		from an exponential distribution of rate~$1$.
	}
	\State{$Y_u \gets (r_{\beta}(\nu_u^{(1)}),
		\ldots, r_{\beta}(\nu_u^{(\ell )}))$
	}
	\State{$d_u \gets \bot$
	}
	\State{$C_u \gets 0$
	}
	\State{\textbf{In each round} \textbf{do}}
	\State{Send $\left \langle C_u,X_u,Y_u\right \rangle$ to all and receive one message from each in-neighbor.
	}
	\If{at least one received message is null}
	\State{$C_u \gets 0$
	}
	\Else{}
	\State{$C_u \gets 1 + \min\set{C_{v_1} , \ldots, C_{v_k} }$ }
	\EndIf{}
	\For{$i = 1 \ldots \ell$}
	\State{$X_u^{(i)} \gets \min \set{X_{v_1}^{(i)}, \ldots, X_{v_k}^{(i)}}$
	}
	\State{$Y_u^{(i)} \gets \min \set{Y_{v_1}^{(i)}, \ldots, Y_{v_k}^{(i)}}$
	}
	\EndFor{}
	\State{$n_u \gets \ell/(Y_u^{(1)} + \cdots + Y_u^{(\ell)})$}
	\If{$C_u > 3 \, n_u/2$}
	\State{$d_u \gets a - 1 + (Y_u^{(1)} + \cdots + Y_u^{(\ell)})/(X_u^{(1)} + \cdots + X_u^{(\ell)}) $
	}
	\EndIf{}
	\end{algorithmic}
\end{algorithm}

\section{Conclusion}\label{sec:conclusion}
The design of average consensus algorithms is constrained by fundamental
	limitations on computable functions.
In a networked system of deterministic agents that communicate by broadcast
	without knowledge of their out-degrees, average consensus essentially requires
	central coordination or global information on the network.
Indeed, although much progress has been made over the past decades, average
	consensus algorithms generally continue to rely on assumptions such as
	bidirectional links, an upper bound on the number of agents known to all
	agents, knowledge of the agent the out-degrees\ldots

Furthermore, most average consensus algorithms are proved correct under the
	condition that agents are able to store and transmit real numbers, which is a
	highly idealized situation.
The above issues hinder the widespread application of many existing average
	consensus algorithms.

We have proposed a Monte Carlo algorithm that achieves average consensus with
	high probability, in linear time, and performs well under limited assumptions
	on the network.
This algorithm can be coupled with a rounding procedure that allows for
	working with quantized values, with a space complexity growing with
	$\log\log\,n$ asymptotically.
In this form, the algorithm only computes the average in the sense that every
	agent converges towards an estimate of the average in finite time.
However, if we provide the agents with an upper bound on the size of the
	network, the algorithm can be augmented in a way that allows the agents to
	eventually decide irrevocably on their estimate.

This method has its own shortcomings: specifically, the restitution of the
	quantized values requires an infinite computing precision.
Nonetheless, the comparison with existing average algorithms is favorable in
	many respects.
In particular, our algorithm converges in linear time in the size of the
	network, tolerates communication channels with finite capacity, and can be
	augmented with irrevocable decisions on the same estimate of the average.
As such, it offers an example of using randomization to circumvent fundamental
	limitations in distributed computing.

\section*{Acknowledgements}
The authors would like to thank Shlomo Moran for fruitful discussions and
  remarks that greatly helped us improve this work.

\bibliography{bibliography}
\bibliographystyle{ieeetr}
\end{document}